\documentclass[runningheads,a4paper]{article}

\usepackage{amssymb,amsmath,amsthm}
\usepackage{color}
\usepackage[mathcal]{euscript}
\usepackage{microtype}
\usepackage{array}
\usepackage{graphicx}
\usepackage[margin=1in]{geometry}
\usepackage{hyperref}

\def\beginaalgo{\begin{minipage}{5in}\vspace{0.1cm}\normalfont\begin{tabbing}
       \quad\=\qquad\=\qquad\=\qquad\=\qquad\=\qquad\=\qquad\=\qquad\=\qquad\=\qquad\=\qquad\=\qquad\=\qquad\=\kill}
\def\endaalgo{\end{tabbing}\vspace{0.1cm}\end{minipage}}
\newenvironment{algorithm}
{\begin{tabular}{|l|}\hline\beginaalgo}
{\endaalgo\\\hline\end{tabular}}

\graphicspath{{./fig/}}

\title{Coresets for $k$-median clustering under Fr\'{e}chet and Hausdorff distances}

\author{Abhinandan Nath\\Mentor Graphics, Fremont, USA\\\texttt{abnath@mentor.com}}

\date{}

\def\eps{\varepsilon}

\newcommand{\etal}{\textit{et~al.}}
\def\dist{\mathtt{d}}
\newcommand{\E}    {\mathbb{E}}
\newcommand{\A}		{{\EuScript{A}}}

\newcommand{\X}		{{\EuScript{X}}}
\newcommand{\C}		{{\EuScript{C}}}

\newcommand{\R}          {\EuScript{R}}

\newcommand{\V}		{{\EuScript{V}}}

\newcommand{\ball}      {{\EuScript{B}}}
\newcommand{\reals}    {\mathbb{R}}

\newcommand{\ptset}   {\zeta}
\newcommand{\traj}      {\gamma}
\newcommand{\abs}[1]{\left\lvert#1\right\rvert}
\newcommand{\norm}[1]{\left\lVert#1\right\rVert}
\newcommand{\cost}   {\textsc{cost}}
\newcommand{\sens}              {\sigma}
\newcommand{\Sens}		{{\EuScript{S}}}
\newcommand{\opt}     {\Delta}
\newcommand{\eqclass}  {\Xi}
\newcommand{\poly} {\text{poly}}
\newcommand{\polylog} {\text{polylog}}

\makeatletter
\renewcommand\subparagraph{%
 \@startsection {subparagraph}{5}{\z@ }{3.25ex \@plus 1ex
 \@minus .2ex}{-1em}{\normalfont \normalsize \bfseries }}%
\makeatother

\newtheorem{theorem}{Theorem}
\newtheorem{corollary}[theorem]{Corollary}
\newtheorem{lemma}[theorem]{Lemma}
\theoremstyle{definition}

\begin{document}

\maketitle

\begin{abstract}
We give algorithms for computing coresets for $(1+\eps)$-approximate $k$-median clustering of polygonal curves (under the discrete and continuous Fr\'{e}chet distance) and point sets (under the Hausdorff distance), when the cluster centers are restricted to be of low complexity. Ours is the first such result, where the size of the coreset is independent of the number of input curves/point sets to be clustered (although it still depends on the maximum complexity of each input object). Specifically, the size of the coreset is $\Theta\left(\frac{k^3lm^{\delta}d}{\eps^2}\log\left( \frac{kl}{\eps}\right)\right)$ for any $\delta > 0$, where $d$ is the ambient dimension, $m$ is the maximum number of points in an input curve/point set, and $l$ is the maximum number of points allowed in a cluster center. We formally characterize a general condition on the restricted space of cluster centers -- this helps us to generalize and apply the importance sampling framework, that was used by Langberg and Schulman for computing coresets for $k$-median clustering of $d$-dimensional points on normed spaces in $\reals^d$, to the problem of clustering curves and point sets using the Fr\'{e}chet and Hausdorff metrics. Roughly, the condition places an upper bound on the number of different combinations of metric balls that the restricted space of cluster centers can hit. We also derive lower bounds on the size of the coreset, given the restriction that the coreset must be a subset of the input objects.
\end{abstract}

\section{Introduction}
\label{sec:intro}
We study coresets for $(1+\eps)$-approximate $k$-median clustering of polygonal curves and finite point sets under the Fr\'{e}chet and Hausdorff distances respectively, under the restriction that the cluster centers have a bounded number of points each. We frame it in the more general context of clustering points in a metric space $\X = (X,\dist)$ where the cluster centers must come from a (possibly infinite) subset $C \subseteq X$ -- this has been dubbed as the $(k,C)$-median problem~\cite{NT20}. We prove a general condition on the subset $C$ which allows us to efficiently compute small-sized $\eps$-coresets using the sensitivity sampling framework of Langberg and Schulman~\cite{LS10}. This gives us the first such coresets for clustering curves and point sets, where the coreset size is independent of the number, but still dependent on the maximum complexity, of input curves/point sets; however the dependence can be made arbitrarily small (Theorem~\ref{thm:coresets_existence}). Specifically, the size of the coreset is $\Theta\left(\frac{k^3lm^{\delta}d}{\eps^2}\log\left( \frac{kl}{\eps}\right)\right)$ for any $\delta > 0$, where $d$ is the ambient dimension, $m$ is the maximum number of points in an input curve/point set, and $l$ is the maximum number of points allowed in a cluster center. As an easy corollary, we are able to give much better bounds on the VC dimension of the dual of the range space (whose ground set is a set of curves / point sets, and the ranges are metric balls under the Fr\'{e}chet and Hausdorff distances) than those obtained by the results of~\cite{DPP19} and the naive exponential bound for the dual range space (Corollary~\ref{cor:vc_dimension}). We also give a lower bound of $\Omega(k/\eps)$ (Theorem~\ref{thm:lower_bound}) on the size of the coresets, provided the coreset must be a subset of the input set.

The $k$-median problem is well-known, where given a set $P$ of $n$ elements from a metric space, we need to select $k$ \emph{centers} to minimize the sum of the distances of each input to its nearest center. Depending on the application, it is important to choose an appropriate distance function. In this paper, the elements to be clustered are polygonal curves and finite point sets in $\reals^d$ for which we use the  Fr\'{e}chet (both continuous and discrete versions) and Hausdorff distances repectively. Clustering is often a first step in data analysis, and the cluster centers can be thought of as representatives for the respective clusters. Often, real-life data is noisy, and restricting the complexity of the cluster centers (e.g., by bounding the number of points per center curve or point set) is an effective way to prevent them from overfitting to the input noise~\cite{NT20, BDGHKLS19, BDS19, DKS16, BDR20}.

With burgeoning data sizes, computational efficiency of clustering has become important. The $k$-median problem is computationally hard to solve exactly, even for the Euclidean metric~\cite{MS84,FG88}. A long line of work on approximation algorithms for the $k$-median problem has developed, with varying running times. The notion of \emph{coresets}, a powerful data-reduction technique, is prominent among these. Vaguely speaking, a coreset is a small-sized proxy for the entire data set so that the cost of any solution on the whole data set and on the coreset are roughly equal. Hence, solving the clustering problem on the coreset (by, say, using a less efficient algorithm) gives us a solution for the entire data set.

\subparagraph{Problem definition.}
A \textbf{metric space} $\X = (X, \dist)$ consists of a set $X$ and a distance function $\dist : X \times X \rightarrow \reals_{\ge 0}$ that satisfies the following conditions -- (i) $\dist(x,x) = 0$ for all $x \in X$; (ii) $\dist(x,y) = \dist(y,x)$ for all $x,y \in X$; and (iii) $\dist(x,z) \le \dist(x,y) + \dist(y,z)$ for all $x,y,z \in X$.

Given $P,C \subseteq X$, the \textbf{$(k,C)$-median} problem is to compute a set $C' \subseteq C$ of $k$ \emph{centers} that minimizes
\begin{align*}
\cost(P,C') = \dfrac{1}{|P|}\sum_{p \in P} \dist(p,C')
\end{align*}
where $\dist(p,C') = \min_{c \in C'} \dist(p,c)$. Here $P$ is finite but $C$ need not be.

The problem can be generalized to include weights $\mu : P \rightarrow \reals_{> 0}$ for each point (without loss of generality, the weights $\mu$ can be normalized to sum up to one), and the cost now becomes
\begin{align*}
\cost(P,C') = \sum_{p \in P} \mu(p) \dist(p,C').
\end{align*}

Given such an instance $P$ of the $(k,C)$-median problem, an \textbf{$\eps$-coreset} (for $\eps > 0$) is a subset $P' \subseteq X$ with weights $\omega : P' \rightarrow \reals_{> 0}$ such that for any $C' \subseteq C$ of size $k$, the following holds
\begin{align*}
\sum_{p \in P'} \omega(p) \dist(p,C') \in (1 \pm \eps) \cost(P,C').
\end{align*}
Hence, we can solve the $(k,C)$-median problem on the (hopefully smaller) $P'$, to get a good solution for $P$.

As special cases, we consider the metric spaces $(T^m_{poly}, \dist_{F}), (T^m_{poly}, \dist_{dF})$, and $(U^m, \dist_H)$, where $T^m_{poly}$ (resp. $U^m$) denotes the set of all polygonal curves (resp. point sets) in $\reals^d$ with at most $m$ points each ($m > 0$), and $\dist_F, \dist_{dF}$, and $\dist_H$ denote the continuous Fr\'{e}chet, discrete Fr\'{e}chet, and Hausdorff distances respectively. Then, the cluster centers must lie in either $T^l_{poly}$ or $U^l$ for $l > 0$, as the case may be.

\subparagraph{Challenges and techniques.}
The problem of $k$-median clustering under Fr\'{e}chet and Hausdorff distances is challenging, not least due to the fact that these metric spaces do not have a constant doubling dimension~\cite{DKS16,NT20}. In this sense, the continuous versions of the distance function (e.g., the continuous Fr\'{e}chet distance) are more difficult than their discrete counterparts. This is because the doubling dimension for the discrete Fr\'{e}chet distance, although not a constant, can be expressed as a function of the maximum complexity of the input curves, whereas the doubling dimension for the continuous Fr\'{e}chet distance is unbounded even when the input curves' complexities are bounded. Also, unlike the discrete Fr\'{e}chet distance, the vertices of a median curve need not be anywhere near a vertex of an input curve, resulting in a huge search space. As far as coresets are concerned, a lot of earlier work on coresets involved various forms of geometric discretization of $\reals^d$~\cite{HPM04,HPK07,Ch09}, something which seems difficult especially for the continuous Fr\'{e}chet distance given the earlier statement about its doubling dimension and the huge search space. 

Our main technique is the sensitivity sampling framework of Langberg and Schulman~\cite{LS10} -- they use it to compute coresets for $k$-median clustering in normed spaces in $\reals^d$. For the framework to be used in our setting, there are two major technical hurdles to overcome. First, showing that the total sensitivity bound of $O(k)$ extends to the more general $(k,C)$-median problem, and that individual senstivities of input points can be estimated using a bicriteria approximation to the $(k,C)$-median problem. Second, generalizing the notion of well-behaved norms from~\cite{LS10} (which bounds the number of cells in any arrangement of centrally symmetric convex sets associated with the norm) to general metric spaces -- to this end, we introduce the concept of \emph{well-behaved subsets} of a metric space. Roughly, such well-behaved subsets can stab only a bounded number of different combinations of metric balls. We show that curves and point sets having small number of points are well-behaved for the Fr\'{e}chet and Hausdorff metrics.

\subparagraph{Related work.} There has been a lot of experimental work in clustering curves and point sets using the Fr\'{e}chet, Hausdorff and other distances, e.g., see~\cite{CWLS11,XZLZ15,HPL15,BDvdLN19,MMR19,BBKNPW20}.

The first theoretical work on clustering curves was done by Driemel \etal~\cite{DKS16}, where they provided $(1+\eps)$-approximation algorithms for the $(k,l)$-median and $(k,l)$-center clustering problems on 1D trajectories using the continuous Fr\'{e}chet distance; they also show that these problems are NP-hard when $k$ is part of the input and $l$ is fixed. The $(k,l)$-clustering problems are identical to the ones we discuss in our paper, i.e., the $k$ center curves can have at most $l$ points each. Buchin, Driemel and Struijs~\cite{BDS19} prove NP-hardness for the $(k,l)$-median problem under the continuous and discrete Fr\'{e}chet distances, and give $(1+\eps)$-approximation algorithms for the $(k,l)$-center and $(k,l)$-median problem under the discrete Fr\'{e}chet distance. Nath and Taylor~\cite{NT20} give much faster $(1+\eps)$-approximation algorithms for the $(k,l)$-median problem under the discrete Fr\'{e}chet and Hausdorff distances -- they use a generalized notion of doubling dimension called $g$-coverability, and it is unknown if this notion extends to the continuous Fr\'{e}chet distance. Buchin, Driemel and Rohde~\cite{BDR20} give a $(1+\eps)$-approximate algorithm for $(k,l)$-median clustering under the continuous Fr\'{e}chet distance, but	 the complexity of the center curves can increase to $2l-2$.

Coresets for the $k$-median problem have been studied extensively in Euclidean spaces. The first such coresets were given by Har-Peled and Mazumdar~\cite{HPM04}; they were of size $O\left(\frac{k}{\eps^d} \log n\right)$, and were later made independent of $n$ by Har-Peled and Kushal~\cite{HPK07} to $O\left(\frac{k^2}{\eps^d}\right)$. Chen~\cite{Ch09} improved the dependence on $d$ for the price of $\log n$, specifically $O\left(\frac{dk}{\eps^2} \log n\right)$. Sohler and Woodruff~\cite{SW18} removed the dependence on $d$ and $n$ entirely by giving coresets of size $\poly(k/\eps)$ using a dimension reduction technique, while Huang and Vishnoi~\cite{HV20} give a coreset of size $O\left( \frac{k}{\eps^4} \polylog(k,1/\eps) \right)$ using two-stage random sampling and terminal embeddings -- it is unclear how to extend these to clustering curves or point sets.

Langberg and Schulman~\cite{LS10} introduced a different, more abstract approach to coreset construction by connecting it with VC-dimension-type results. This is done by associating a function with each $k$-tuple of points, defined on every point of the metric space, and the cost of clustering with any $k$-tuple as cluster centers is the sum of the values attained by the corresponding function on the input points -- the final coreset size is $O\left(\frac{d^2k^3}{\eps^2} \polylog(d,k,1/\eps)\right)$. Feldman and Langberg~\cite{FL11} introduce a different set of functions, one for each input point for a total of $n$ functions, and defined on every set of $k$ centers, and the cost of clustering with a set of $k$ centers is the sum of the $n$ function values attained on that particular center set -- the final coreset size is $O\left(\frac{dk}{\eps^2}\right)$. However, this requires a bound on the VC/shattering dimension of the range space induced by \emph{weighted} distance functions (where ranges are of the form $\{x \in X | w(x)\cdot \dist(x,y) \le r\}$ for $r > 0, y \in X$ and weights $w$), something which is difficult even for doubling metric spaces~\cite{HJLW18}. Huang \etal~\cite{HJLW18} introduce a probablistic variant of the shattering dimension for a slightly perturbed distance function that gives small coresets for doubling metric spaces -- it will be interesting to see if this approach can be extended to the Fr\'{e}chet and Hausdorff distances. Buchin and Rohde~\cite{BR19} give coresets of size $O(\log n / \eps^2)$ for the discrete $k$-median problem (i.e., the cluster centers come from the input set) under the continuous Fr\'{e}chet distance. For general metric spaces, a lower bound of $\Omega\left( \frac{k}{\eps} \cdot tw \right)$ and an almost tight upper bound is known~\cite{BBHJKW20}, where $tw$ is the treewidth of the metric graph and can be $\Theta(\log n)$ in the worst case.

The rest of the paper is organized as follows. We give the necessary technical background in Section~\ref{sec:prelim}. We give sensitivity bounds for the $(k,C)$-median problem, and discuss how to estimate sensitivities of individual points using a bicriteria approximation in Section~\ref{sec:sensitivity}. In Section~\ref{sec:simplicity_cover}, we define the notion of well-behaved subsets, and show how it applies to bounded complexity curves and point sets. We then tie this with the existence of small coresets using the idea of cover codes~\cite{LS10}. In Section~\ref{sec:computing_coresets} we give the overall algorithm for computing coresets. Section~\ref{sec:lower_bound} gives lower bounds on the coreset size.

\section{Preliminaries}
\label{sec:prelim}

\subparagraph{Comparing curves and finite point sets.}
In its most general form, a \emph{curve} in $\reals^d$ can be specified using a continuous parameterization $\traj : [0,1] \rightarrow \reals^d$. A \emph{polygonal curve} can also be specified using a finite sequence of points $\langle p_1, p_2, \ldots \rangle$ in $\reals^d$ called \emph{vertices}, with each consecutive pair of vertices being joined by a segment $\overline{p_i p_{i+1}}$ called an \emph{edge}.

Given two parameterized curves $\traj_1$ and $\traj_2$, the \textbf{Fr\'{e}chet distance} between them is
\begin{align*}
\dist_{F}(\traj_1, \traj_2) = \min_{f, g : [0,1] \rightarrow [0,1]} \max_{\alpha \in [0,1]} \norm{\traj_1(f(\alpha)) - \traj_2(g(\alpha))}
\end{align*}
where $\norm{.}$ denotes the $l_2$ norm, and $f,g$ range over all continuous, non-decreasing functions with $f(0) = g(0) = 0$ and $f(1) = g(1) = 1$. Intuitively, imagine a man on one curve walking a dog on the other while holding a fixed length leash, with both of them starting and ending at the start and end points of the respective curves, and none of them allowed to move backwards. The Fr\'{e}chet distance gives the length of the shortest leash that makes such a walk possible\footnote{Strictly speaking, $\dist_F$ is a pseudometric, i.e., we can have $\dist_F(\traj_1, \traj_2) = 0$ for $\traj_1 \ne \traj_2$. This technicality however does not affect our results in any way.}.

The \textbf{discrete Fr\'{e}chet distance} is defined for polygonal curves, and it only takes into account the vertices while disregarding the edges. Given two such curves $\traj_1 = \langle p_1, \ldots \rangle$ and $\traj_2 = \langle q_1, \ldots \rangle$, a correspondence is a subset of $\{p_1, \ldots\} \times \{q_1, \ldots \}$ such that every vertex appears in at least one pair. Such a correspondence $\C$ is said to be monotone iff for all $(i_1,j_1), (i_2,j_2) \in \C, i_2 \ge i_1 \Rightarrow j_2 \ge j_1$. The discrete Fr\'{e}chet distance $\dist_{dF}$ between $\traj_1$ and $\traj_2$ is then defined as
\begin{align*}
\dist_{dF} (\traj_1, \traj_2) = \min_{\C} \max_{(p,q) \in \C} \norm{p-q}
\end{align*}
where $\C$ ranges over all \emph{monotone} correspondences between $\traj_1$ and $\traj_2$.

Let $T^m_{poly}$ denote the set of all polygonal curves in $\reals^d$ with at most $m$ vertices; thus $T_{poly} = \bigcup_{i \ge 1} T^i_{poly}$ is the set of all polygonal curves.

Given two finite point sets $\ptset_1 = \{p_1, \ldots\}$ and $\ptset_2 = \{q_1, \ldots\}$ in $\reals^d$, the \textbf{Hausdorff distance} between them is defined as
\begin{align*}
\dist_{H}(\ptset_1, \ptset_2) = \min_{\C} \max_{(p,q) \in \C} \norm{p-q}
\end{align*}
where $\C$ ranges over all correspondences between $\ptset_1$ and $\ptset_2$.

\subparagraph{Range space, VC and shattering dimension.}
A \textbf{range space} $(X, \R)$ consists of a ground set $X$ and a collection of ranges $\R$, where each range $R$ is a subset of $X$. For $Y \subseteq X$, let
\begin{align*}
\R_{|Y} = \{ R \cap Y | R \in \R\}.
\end{align*}
Then, $Y$ is said to be \emph{shattered} by $\R$ if $\R_{|Y}$ contains all subsets of $Y$. The \textbf{Vapnik-Chervonenkis (or VC) dimension}~\cite{VC15} of $(X,\R)$ is the size of the largest shattered subset of $X$. The \textbf{shattering dimension} of $(X, \R)$ is the smallest $\delta$ such that for all $m$, 
\begin{align*}
\max_{Y \subset X, |Y| = m} | \R_{|Y} | = O(m^\delta).
\end{align*}
For a range space with VC-dimension $\nu$ and shattering dimension $\delta$, $\nu = O(\delta \log \delta)$ and $\delta = O(\nu)$~\cite{HP11}.

Given a range space $(X,\R)$, for any $p \in X$, let $\R_p = \{R \in \R \mid p \in R\}$ be the set of ranges containing $p$. The \textbf{dual range space} of $(X,\R)$ is the range space $(\R, \{\R_p \mid p \in X\})$. If a range space has VC-dimension $\nu$, its dual range space has VC-dimension $\le 2^{\nu + 1}$~\cite{HP11}.

\subparagraph{Sensitivity-sampling framework.}
We briefly describe the sensitivity-sampling framework introduced by Langberg and Schulman~\cite{LS10} in the more general setting of approximating integrals of functions.

Let $F$ be a non-negative real-valued family of functions defined on some set $X$, and let $\mu$ be a probability distribution on $X$. The goal is to approximate $\overline{f} = \int_X f(x) d\mu$ for all $f \in F$, to within a factor of $(1 \pm \eps)$ for some small $\eps > 0$, using a finite subset $R$ of $X$ with positive weights $\omega$, i.e., $\sum_{x \in R} \omega(x) f(x) \in (1 \pm \eps) \overline{f}$. Such an $R$ is called an \textbf{$\eps$-approximator} of $F$. We next discuss how importance sampling can be used to compute $R$.

For any $x \in X$, its \textbf{sensitivity} is defined as $\sens_{F,\mu}(x) = \sup_{f \in F} f(x)/\overline{f}$; we often drop the subscript from $\sens_{F,\mu}$ if it is clear from the context. The \textbf{total sensitivity} of $F$ is then defined as $\Sens(F) = \sup_{\mu} \int \sens_{F,\mu} (x) d\mu$. Further, let $s_{F,\mu}(x)$ (or $s(x)$) be an upper bound on $\sens_{F,\mu}(x)$ (or $\sens(x)$) for all $x \in X$, and let $S(F) = \sup_{\mu} \int s_{F,\mu}(x) d\mu(x)$.

Consider the probability distribution $q(x) = s(x) \mu(x)/S(F)$ to be used for importance sampling while estimating $\overline{f}$. Then $G = f(x) \mu(x)/q(x)$ is an unbiased estimator for $\overline{f}$, i.e., $\E_q[G] = \overline{f}$, and its variance can be bounded as follows.
\begin{lemma}[Theorem 2.1, \cite{LS10}]
\label{lem:variance}
$\textsc{Var}(G) \le (S(F) -1) \overline{f}^2.$
\end{lemma}
The following lemma then gives a concentration bound in terms of the number of independent samples needed from $q(x)$ for estimating $\overline{f}$, upto a multiplicative factor.
\begin{lemma}[Lemma 2.1, \cite{LS10}]
\label{lem:conc}
Let $\eps > 0, f \in F$, and $R$ be a random sample of $X$ of size $a \ge \frac{2(S-1)}{\eps^2}$ drawn according to the distribution $q$. Then
\begin{align*}
\Pr \left[ \left| \overline{f} - \sum_{x \in R} \left(\tfrac{S(F)}{a \cdot s(x)}\right) f(x) \right| \ge \eps \overline{f} \right] \le 1/2.
\end{align*}
\end{lemma}

Note that the above bound holds only for a fixed $f \in F$ -- if we want it to hold for all $f \in F$, we will have to use additional properties of $F$ and $X$ (similar to the VC-dimension arguments that have been used earlier~\cite{HW87}). We later show in our paper how to accomplish this for the $(k,C)$-median problem.

\section{Bounded sensitivity for $(k,C)$-median clustering}
\label{sec:sensitivity}

We frame the $(k,C)$-median clustering problem in terms of functions in order to apply the sensitivity-sampling framework.
For any $k$-subset $\{c_1, c_2, \ldots, c_k\} \subseteq C$, define $f_{c_1, \ldots, c_k} : X \rightarrow \reals_{\ge 0}$ as
\begin{align*}
f_{c_1,\ldots, c_k} (x) = \min_{c \in \{c_1, \ldots, c_k\}} \dist(x,c).
\end{align*}

Then for any distribution $\mu$ with finite support P in $X$, we have 
$
\cost(P, \{c_1,\ldots,c_k\}) = \overline{f}_{c_1,\ldots,c_k}
$. Thus, the $(k,C)$-median problem defines a function family 
$
F_{k,C} = \{ f_{c_1, \ldots, c_k} \mid c_1, \ldots, c_k \in C \}
$.
This allows us to talk about sensitivities in the $(k,C)$-median setting\footnote{While the $(k,C)$-median problem is defined for a distribution with finite support $P$, the arguments in this section also work for any distribution $\mu$ on $X$. Hence we use integration, which in the case of a finite support becomes a summation.}. The next lemma bounds the sensitivity in this setting.

\begin{lemma}
\label{lem:sensitivity_metric_space}
For any metric space $\X = (X,\dist)$ and $C \subseteq X$, the total sensitivity of $F_{k,C}$ is $\Sens(F_{k,C}) \le 4k + 6$.
\end{lemma}
The proof of this lemma is similar to that of the lemma below, and is based on the proof of~\cite[Theorem 3.1]{LS10}. The main difference is that we are now in the more general setting of the $(k,C)$-median problem. See Section~\ref{sec:sensitivity_proof} in the Appendix for the full proof.

\subparagraph{Sensitivity upper bound via bicriteria approximation.}
We show how to compute  upper bounds on the individual sensitivity of each point in $X$, while still making sure that the corresponding upper bound on the total sensitivity of $F_{k,C}$ is not too big, since the size of our coreset will depend on it.

Given $\alpha, \beta \ge 1$, an $(\alpha, \beta)$-approximation to an instance of the $(k,C)$-median problem is a set of at most $\beta k$ centers from $C$ such that cost of clustering using these centers is at most $\alpha$ times the cost of the optimal clustering. Hence, the function $f$ associated with these centers satisfies
\begin{align*}
\overline{f} = \int_X f d\mu \le \alpha \opt
\end{align*}
where $\opt$ is the optimal cost of the $(k,C)$-median problem.

\begin{lemma}
\label{lem:bicriteria}
Given an $(\alpha, \beta)$-approximation to an instance of the $(k,C)$-median problem, we can compute for all $x \in X$, a value $s(x) \ge \sens(x)$ satisfying $S = \int_X s(x) d\mu \le 6\alpha + 4\beta k$.
\end{lemma}
\begin{proof}
Let $f = f_{\{c_1, \ldots, c_{\beta k}\}}$ be the function associated with the bicriteria approximation, i.e., $\bar{f} \le \alpha \opt$. Let $\bar{f} = \opt'$, and let $\V(c_i)$ be the Voronoi cell of $c_i$. Further, let $\mu_i = \mu(\V(c_i))$, and $m_i = \frac{1}{\mu_i} \int_{\V(c_i)} \dist(x,c_i)d\mu$, so that $\opt' = \sum_i \mu_i m_i$. By Markov's inequality, for each $i$, $\mu\left(\ball(c_i, 2m_i) \cap \V(c_i)\right) \ge \mu_i/2$.

Next, we show that
\begin{align*}
s(x) = \frac{2\alpha(2m_i + \dist(x,c_i))}{\opt'} + \frac{4}{\mu_i} \ge \sens(x).
\end{align*}

Let $x \in \V(c_i)$, and let $f' = f_{\{c'_1, \ldots, c'_k\}}$ for any set of $k$ points $\{c'_1, \ldots, c'_k\} \subseteq C$. Further, let $c'$ denote the closest point to $c_i$ in $\{c'_1, \ldots, c'_k\}$, and let $\dist_i = \dist(c', c_i)$. Then by triangle inequality and the preceding inequality
\begin{align*}
\bar{f'} = \int\limits_{X} f' d\mu
\ge \int\limits_{\V(c_i)} f' d\mu
\ge  \int\limits_{\ball(c_i,2m_i) \cap \V(c_i)} f' d\mu
\ge \max \{ 0, \dist_i - 2m_i\} \mu_i/2.
\end{align*}
Also, $\bar{f'} \ge \opt'/\alpha$. Hence, $\bar{f'} \ge \max\{0, \dist_i - 2m_i\} \mu_i/4 + \opt'/2\alpha$. We thus have
\begin{align*}
\sens(x)
&= \max_{f'} f'(x)/\bar{f'}\\
&= \max_{f'} \dist(x, c') / \bar{f'}\\
&\le \max_{f'} \frac{\dist_i + \dist(x,c_i)}{\max\{0, \dist_i - 2m_i\} \mu_i/4 + \opt'/2\alpha}\\
&\le \max_{\dist_i \ge 2m_i} \frac{\dist_i + \dist(x,c_i)}{(\dist_i - 2m_i) \mu_i/4 + \opt'/2\alpha}.
\end{align*}
The right hand side is maximized either at $\dist_i = 2m_i$ or $\dist_i = \infty$. We conclude that
\begin{align*}
\sens(x) \le \max\left\{\frac{2\alpha(2m_i + \dist(x,c_i))}{\opt'}, \frac{4}{\mu_i}\right\} \le\frac{2\alpha(2m_i + \dist(x,c_i))}{\opt'} + \frac{4}{\mu_i}.
\end{align*}

Further,
\begin{align*}
S
&= \int\limits_X s(x) d\mu = \sum_i \int\limits_{\V(c_i)} s(x) d\mu \\
&\le \sum_i \int\limits_{\V(c_i)} \left( \frac{4\alpha m_i + 2\alpha \dist(x,c_i)}{\opt'} + \frac{4}{\mu_i}\right) d\mu \\
&= \sum_i \left(  \frac{4 \alpha m_i \mu_i}{\opt'} + \frac{2\alpha m_i \mu_i}{\opt'}+ 4\right)\\
&\le 6 \alpha + 4 \beta k.
\end{align*}
\end{proof}
\section{Well-behaved subsets and cover codes}
\label{sec:simplicity_cover}
We formally characterize the subset $C$ that allows us to compute small coresets for clustering.

\subparagraph{Well-behaved subsets.}
Given a (possibly infinite) collection of sets $\A = \{A_1, A_2, \ldots \}$ from some universe $U$, we define the following equivalence relation $\sim$ on $U$ induced by $\A$ : we say $x \sim y$ whenever $x,y$ lie in the same sets in $\A$, i.e, 
\begin{align*}
x \sim y \Leftrightarrow \{i \mid x \in A_i\} = \{i \mid y \in A_i\}.
\end{align*}
We denote the set of equivalence classes by $\eqclass_{\A}^U$; note that the equivalence classes partition the universe $U$.

Let $h : \reals_{\ge 0} \rightarrow \reals_{\ge 0}$ be a non-decreasing function. Given a metric space $\X = (X, \dist)$, a metric ball $B(p,r)$ (for some $p \in X, r > 0$) is defined as $B(p,r) = \{x \in X \mid \dist(p,x) \le r\}$. Then, $\X$ is said to be \textbf{$h$-well-behaved} w.r.t. a (possibly infinite) subset $C \subseteq X$, iff for any collection of $n$ sets $\ball = \{B_1, \ldots, B_n\}$ in $X$ (where each $B_i$ is a Boolean combination, i.e., obtained by union, intersection and complement, of a constant number of metric balls of arbitrary centers and radii in $\X$),
$C$ has a non-empty intersection with at most $h(n)$ equivalence classes from $\eqclass_{\ball}^X$.
Roughly, it is a measure of the simplicity of $C$, in that it upper bounds the total number of different combinations of sets in $\ball$ that can be simultaneously intersected by $C$.

\begin{figure}[h]
\centering
\includegraphics[width=\textwidth]{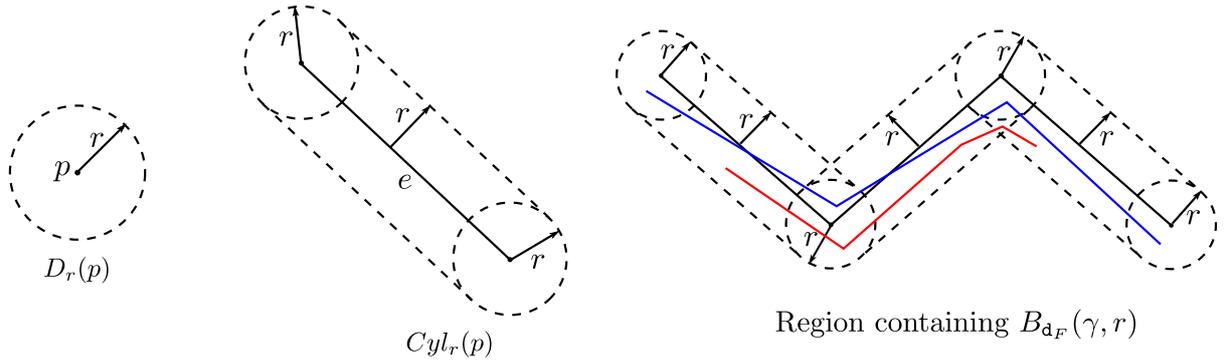}
\caption{Basic shapes $D_r(p)$ and $Cyl_r(p)$, and the region containing the Fr\'{e}chet ball $B_{\dist_F}(\traj,r)$ for the solid black curve $\traj$ with four vertices. The shapes and ball are the union of the regions with dotted boundary. Not all curves that lie in the region for $B_{\dist_F}(\traj,r)$ have Fr\'{e}chet distance to $\traj$ at most $r$, but any curve in $B_{\dist_F}(\traj,r)$ must lie inside the dotted region (e.g., the blue curve is in $B_{\dist_F}(\traj,r)$ but the red curve is not).}
\label{fig:frechet_ball}
\end{figure}

\begin{lemma}
\label{lem:cfd_well_behaved}
$(T^m_{poly}, \dist_F)$ is $h$-well-behaved w.r.t. $T^l_{poly}$, where $h(n) = (cnm/d)^{dl}$ for some constant $c > 0$.
\end{lemma}
\begin{proof}
 For a given $r > 0$ and $\traj = \langle p_1, \ldots, p_m \rangle$, the Fr\'{e}chet metric ball of radius $r$ centered at $\traj$, denoted by $B_{\dist_F} (\traj,r)$, is contained in the union of the following \emph{basic} shapes in $\reals^d$ (see Fig.~\ref{fig:frechet_ball})

(i) $m$ disks $D_r(p) = \{ x \in \reals^d \mid \norm{(x-p)} \le r \}$ of radius $r$ centered at each vertex $p$ of $\traj$, and 

(ii) $m-1$  closed Euclidean balls $Cyl_r(e) = \{ x \in \reals^d \mid \exists y \in e \text{ s.t. } \norm{x-y} \le r\}$ centered along each edge $e$ of $\traj$. 

Hence, the \emph{points} of any curve in a metric ball in $(T^m_{poly}, \dist_F)$ must lie in the union of at most $2m-1$ such  shapes (see Fig.~\ref{fig:frechet_ball}). The argument can be extended for any curve lying in the Boolean combination of a constant number of metric balls in $(T^m_{poly}, \dist_F)$, in which case the curve's points must lie in a region of $\reals^d$ defined by $O(m)$ basic shapes (provided the Boolean output set is non-empty). The boundary of each such basic shape is, in turn, given by the zero set of a constant number of polynomials in $\reals^d$ of degree at most two. Thus, the boundary of the Boolean combination of a constant number of metric balls in $(T^m_{poly}, \dist_F)$ can be defined using the zero sets of $O(m)$ polynomials in $\reals^d$ of degree at most two (note that the output of the Boolean operations forms a connected set in $\reals^d$); see Fig.~\ref{fig:frechet_boolean}.

\begin{figure}[h]
\centering
\includegraphics[width=\textwidth]{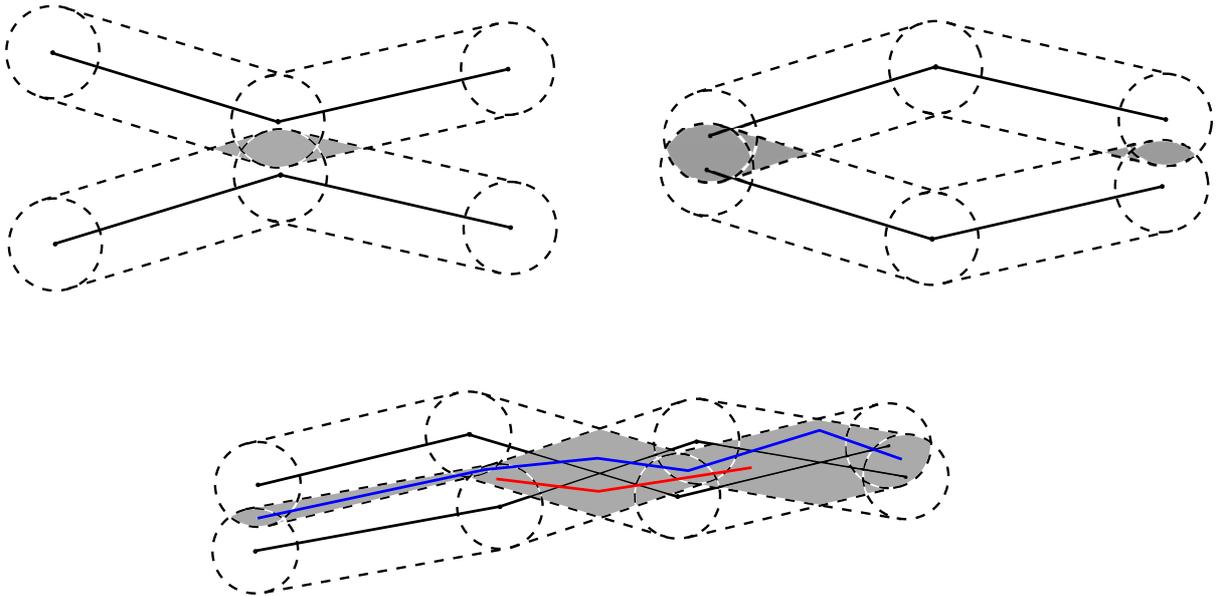}
\caption{Pairwise intersection of Fr\'{e}chet balls. Center curves are solid black, and the ball regions have dotted boundary. Even though the two ball regions in each of the top two figures have non-empty intersection (shown in grey), the intersection of the two Fr\'{e}chet balls is empty (i.e., has no curves). Any curve in the intersection of two Fr\'{e}chet balls must lie in the intersection between the regions containing the two Fr\'{e}chet balls, but not vice versa. e.g., in the bottom figure, both the blue and red curves lie in the intersection of the ball regions (in grey), but only the former is in the Fr\'{e}chet ball of both curves, while the latter is in none.}
\label{fig:frechet_boolean}
\end{figure}

Consider $n$ sets $\ball = \{B_1, \ldots, B_n\}$, each $B_i \subseteq T^m_{poly}$ being the Boolean combination of a  constant number of metric balls in $(T^m_{poly}, \dist_F)$. From the discussion in the preceding paragraph, the boundary of the sets in $\ball$ can be defined using zero sets of $O(nm)$ polynomials in $\reals^d$ of degree at most two. Consider the arrangement $\EuScript{Z}$ defined by these zero sets -- the number of cells in the arrangement is $(cnm/d)^d$ for some constant $c > 0$ (from~\cite[Theorem 6.2.1]{Mat13}, see also~\cite[Corollary 5.6]{Mat09}). Then, a curve $\traj' \in T^l_{poly}$ that lies in an equivalence class of $\eqclass_{\ball}^{T^m_{poly}}$ must have all its points contained completely in $\cup_i B_i$, and hence can be identified by the at most $l$ cells of  the arrangement $\EuScript{Z}$ that its vertices lie in (see Fig.~\ref{fig:frechet_boolean}). By a counting argument, the number of such curves $\traj'$, and thereby the number of equivalence classes of $\eqclass_{\ball}^{T^m_{poly}}$ having non-empty intersection with $T^l_{poly}$, is at most $(cnm/d)^{dl}$.
\end{proof}

\begin{lemma}
\label{lem:dfd_well_behaved}
$(T^m_{poly}, \dist_{dF})$ is $h$-well-behaved w.r.t. $T^l_{poly}$, where $h(n) = (cnm/d)^{dl}$ for some constant $c > 0$.
\end{lemma}
\begin{proof}
Given curve $\traj = \langle p_1, \ldots, p_m \rangle$, any curve $\traj'$ with $\dist_{dF}(\traj,\traj') \le r$ for some $r > 0$ must have each of its vertices in one of the $m$ disks of radius $r$ centered at vertices of $\traj$ (note that such a metric ball in $(T^m_{poly}, \dist_{dF})$ need not form a connected set in $\reals^d$) (see Fig.~\ref{fig:discrete_frechet}).
 
\begin{figure}[h]
\centering
\includegraphics[width=\textwidth]{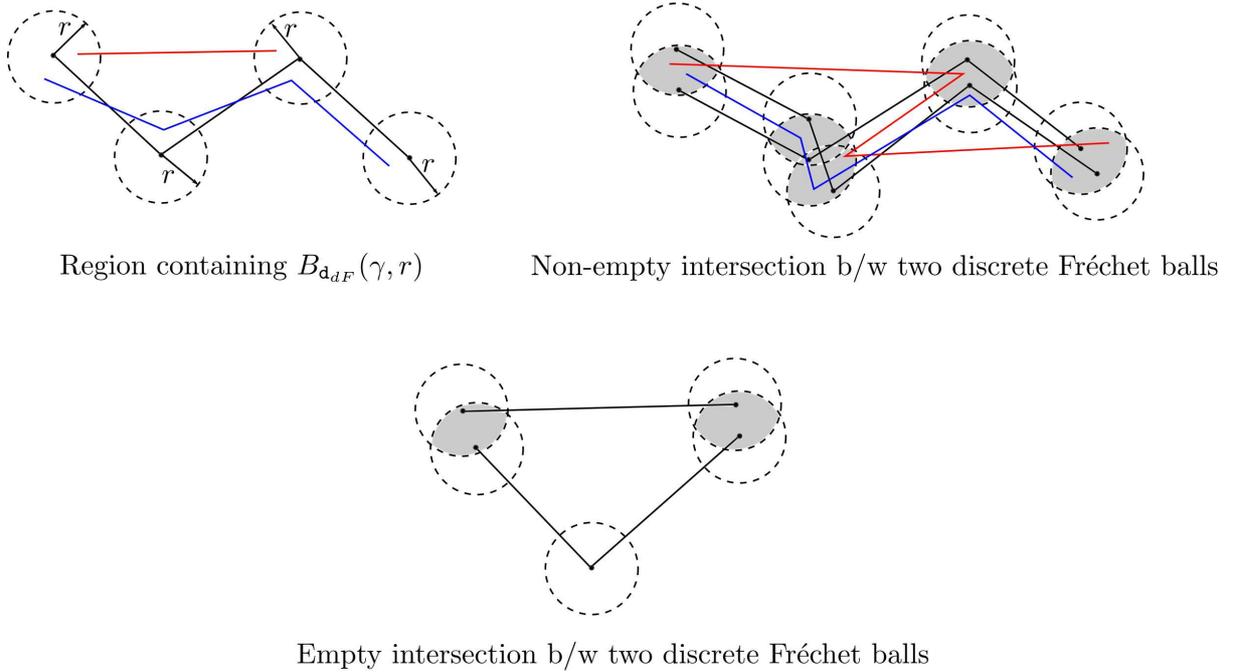}
\caption{First figure in top row shows the disks centered at vertices of the black curve $\traj$ -- any curve in $B_{\dist_{dF}}(\traj,r)$ must have its vertices inside these disks, but not vice versa (e.g., the blue curve is in $B_{\dist_{dF}}(\traj,r)$ but the red curve is not). Second figure in top row shows the intersection between two discrete Fr\'{e}chet balls -- any curve in this intersection must have its vertices in the grey regions, but not vice versa (e.g., the blue curve is in the intersection, the red curve is not). The figure on the bottom row shows two discrete Fr\'{e}chet balls with no curve in their intersection, even though the disks have non-zero intersection (the grey regions).}
\label{fig:discrete_frechet}
\end{figure}

Consider $n$ sets $\ball = \{B_1, \ldots, B_n\}$  where each $B_i \subseteq T^m_{poly}$ is a Boolean combination of a constant number of metric balls in $(T^m_{poly}, \dist_{dF})$. From the previous paragraph and the discussions in proof of Lemma~\ref{lem:cfd_well_behaved}, the boundary of each $B_i$ is defined by zero sets of $O(m)$ polynomials of degree at most two (see Fig.~\ref{fig:discrete_frechet}). Consider the arrangement $\EuScript{Z}$ formed by the zero sets of the $O(nm)$ polynomials for all sets in $\ball$ -- the number of cells in the arrangement $\EuScript{Z}$ is at most $(cnm/d)^d$ for some constant $c > 0$ (see~\cite[Theorem 6.2.1]{Mat13}, see also~\cite[Corollary 5.6]{Mat09}). Any curve $\traj' \in T^l_{poly}$ that lies in an equivalence class of $\eqclass_{\ball}^{T^m_{poly}}$ must have all its vertices in $\EuScript{Z}$, and  hence can be identified by the at most $l$ cells of the arrangement $\EuScript{Z}$ that its vertices lie in. Thus there are at most $(cnm/d)^{dl}$ such curves, which also bounds the number of equivalence classes of $\eqclass_{\ball}^{T^m_{poly}}$ having non-empty intersection with $T^l_{poly}$.
\end{proof}

Let $U^m$ denote the set of all finite point sets with at most $m$ points each. The proof of the following lemma is similar to that of Lemma~\ref{lem:dfd_well_behaved}, and is omitted for brevity.
\begin{lemma}
\label{lem:hd_well_behaved}
$(U^m, \dist_H)$ is $h$-well-behaved w.r.t. $U^l$, where $h(n) = (cnm/d)^{dl}$ for some constant $c > 0$.
\end{lemma}

\subparagraph{Remark on VC-dimension.} Consider the range space with ground set $T^l_{poly}$ and ranges $\R_{\dist_F, l, m} = \{ B_{\dist_F}(\traj, r) \cap T^l_{poly} \mid r > 0, \traj \in T^m_{poly}\}$ defined by Fr\'{e}chet metric balls centered on curves in $T^m_{poly}$. The VC-dimension of this range space is $O(d^2m^2 \log(dlm))$~\cite{DPP19}. We can similarly define the range spaces $(T^l_{poly}, \R_{\dist_F, l, m})$ and $(U^l_{poly}, \R_{\dist_H, l, m})$ -- these have VC-dimension $O(dm\log(dlm))$~\cite{DPP19}.
Then, the quantity $h(n)$ gives an upper bound on the number of shattered subsets of a ground set of size $n$ for the dual range space of the preceding range spaces. This gives the following bound on the shattering dimension of the dual range space (and the VC dimension upto logarithmic factors), which is much better than those that can be obtained using the naive exponential bound and the results of~\cite{DPP19}.
\begin{corollary}
\label{cor:vc_dimension}
The shattering dimension of the dual of the range spaces $(T^l_{poly}, \R_{\dist_F, l, m})$, $(T^l_{poly}, \R_{\dist_{dF}, l, m})$, and $(U^l_{poly}, \R_{\dist_H, l, m})$ is $O(dl \log (m/d))$.
\end{corollary}

\subparagraph{Small cover codes.}
We discuss the notion of \emph{cover codes} as introduced in \cite{LS10}. Intuitively, a cover code is a subset of functions from the function family $F$ that approximates the set $F$ with respect to a finite subset of the support in $X$. A small-sized cover code plays a crucial role in sidestepping a naive union bound while extending Lemma~\ref{lem:conc} for all functions in F, analogous to the role that a bounded VC-dimension plays, e.g., in~\cite{HW87}.

We now formally define cover codes. Let $A \subseteq X$ be a subset with $a$ elements. For any $g : X \rightarrow \reals$, let $\nu_A(g) = (1/a) \sum_{x \in A} g(x)$. For $f,f' \in F$ and $x \in A$, we define $\hat{f} = \nu_A(f/s)$ (where $s$ is an upper bound on sensitivity) and
\begin{align*}
D_{A,x} (f,f') = \left| \frac{f(x)}{\hat{f} \cdot s(x)} - \frac{f'(x)}{\hat{f'} \cdot s(x)} \right| .
\end{align*}

Let $S$ be the upper bound on total sensitivity computed using $s$. Then, $F' \subseteq F$ is an \textbf{\emph{$\eps$-cover-code}} for $(F,A,s)$ for some $\eps > 0$ iff for all $f \in F$ there exists an $f' \in F'$ such that
\begin{enumerate}
\item[(i)] $\frac{\overline{f'}}{\hat{f'}} \le \frac{\overline{f}}{\hat{f}}$, and
\item[(ii)] For all $x \in A$, $D_{A,x}(f,f') \le \frac{\eps}{64S} \left( 1 + \frac{f(x)}{\hat{f} \cdot s(x)} + \frac{f'(x)}{\hat{f'} \cdot s(x)} \right)$.
\end{enumerate}

Coming back to our setting of the $(k,C)$-median problem, the following theorem is the main result of this section, and states that if $C$ is simple enough, then $F_{k,C}$ has small cover codes.

\begin{theorem}
\label{thm:cov}
Let $A \subseteq X$ be of size $a$. If $\X$ is $h$-well-behaved w.r.t. $C$, then $F_{k,C}$ has an $\eps$-cover-code of size $\left(h\left(\left( \frac{Sa}{\eps}\right)^{\Theta(1)}\right)\right)^k$ for any $\eps > 0$.
\end{theorem}

We will prove this in two steps. First, we show that there exists a small set of functions $G$ (not necessarily in $F_{k,C}$) of size $\left(h\left(\left( \frac{Sa}{\eps}\right)^{\Theta(1)}\right)\right)^k$ such that for any $f \in F_{k,C}$, there exists a constant $c_f$ and a function $g \in G$ such that for any $x \in A$,
\begin{align}
\label{eq:g_req}
\left| \frac{f(x)}{s(x)} - \frac{c_f g(x)}{s(x)} \right| \le \frac{\eps}{256S}\hat{f}.
\end{align}
Second, $G$ does not quite give us the desired cover, since it may not lie in $F_{k,C}$. However, it can be shown that there exists $f_g \in F$ for each $g \in G$ such that the set $\{f_g \mid g \in G\}$ gives us the desired cover~\cite[Lemma 7.1]{LS10}.

We now show the existence of $G$; this will also prove Theorem~\ref{thm:cov}. The proof is rather technical, and borrows heavily from the proof of~\cite[Theorem 4.3]{LS10}. The important difference is an ingenious use of the $h$-well-behaved property of $\X$ w.r.t. $C$ to make sure the proof goes through in the $(k,C)$-median setting. See Section~\ref{sec:proof_cover} for the full proof.
\section{Computing coresets}
\label{sec:computing_coresets}

\begin{figure}[t]
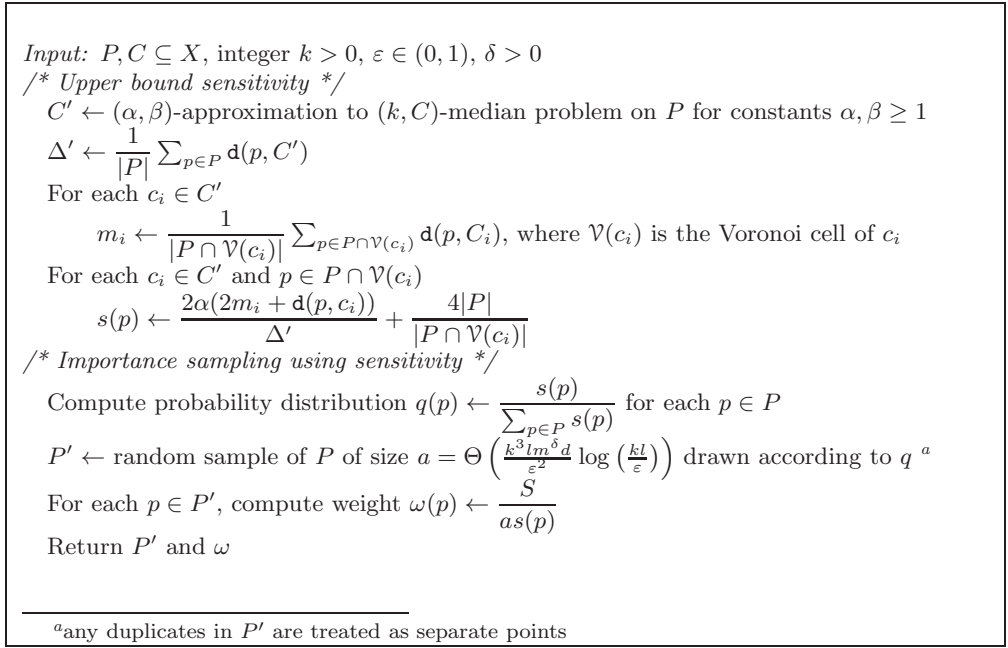

\centering\small
\begin{algorithm}
\\ \textit{Input:} $P, C \subseteq X$, integer $k > 0$, $\eps \in (0,1)$, $\delta > 0$
\\ \textit{/* Upper bound sensitivity */}\+
\\ $C' \leftarrow (\alpha, \beta)$-approximation to $(k,C)$-median problem on $P$ for constants $\alpha, \beta \ge 1$
\\ $\opt' \leftarrow \dfrac{1}{|P|}\sum_{p \in P} \dist(p, C')$
\\ For each $c_i \in C'$ \+
\\ $m_i \leftarrow \dfrac{1}{|P \cap \V(c_i)|} \sum_{p \in P \cap \V(c_i)} \dist(p,C_i)$, where $\V(c_i)$ is the Voronoi cell of $c_i$ \-
\\ For each $c_i \in C'$ and $p \in P \cap \V(c_i)$ \+
\\ $s(p) \leftarrow \dfrac{2\alpha(2m_i + \dist(p, c_i))}{\opt'} + \dfrac{4 |P|}{|P \cap \V(c_i)|}$\-\-
\\ \textit{/* Importance sampling using sensitivity */}\+
\\ Compute probability distribution $q(p) \leftarrow \dfrac{s(p)}{\sum_{p \in P} s(p)}$ for each $p \in P$
\\ $P' \leftarrow$ random sample of $P$ of size $a = \Theta\left(\frac{k^3lm^{\delta}d}{\eps^2}\log\left( \frac{kl}{\eps}\right)\right)$ drawn according to $q$ \footnote{any duplicates in $P'$ are treated as separate points}
\\ For each $p \in P'$, compute weight $\omega(p) \leftarrow \dfrac{S}{as(p)}$
\\ Return $P'$ and $\omega$
\end{algorithm}
\caption{Overall algorithm for computing coresets for $(k,C)$-median under Fr\'{e}chet and Hausdorff distances.}
\label{fig:alg}
\end{figure}

We give the overall algorithm for computing coresets in Fig.~\ref{fig:alg} for Fr\'{e}chet and Hausdorff distances.

Briefly, given input $P$, it first uses a bicriteria approximation to the $(k,C)$-median problem to upper bound the sensitivity $s(p)$ for each input point $p \in P$ (see Lemma~\ref{lem:bicriteria} and its proof).
It then performs importance sampling according to the distribution $q$. In order to bound the number of samples needed (and hence the overall size of the coreset), the following theorem is useful. It is akin to a general VC-type result, and states that small cover codes lead to succinct coresets using random sampling. The proof, in fact, uses a double sampling argument similar to the proof for small $\eps$-nets~\cite{HW87}, and holds for any function family $F$ with domain $X$ and total sensitivity at most $S$.

\begin{theorem}[\cite{LS10}, Theorem 4.4]
\label{thm:cover_coreset}
Suppose for some $a \ge 8(S-1)/\eps^2$, every $A \subseteq X$ of cardinality $|A| = 2a$ has an $\eps$-cover-code (w.r.t. $F$ and sensitivity bound $S$) of cardinality at most $\frac{1}{8}e^{\frac{a \eps^2}{100S^2}}$. Then, a sample of $a$ points from the probability distribution $q$ (with weight of each point $p$ being $S/as(p)$) is an $\eps$-approximator for $F$, with probability $\ge 1/2$.
\end{theorem}
For $(k,C)$-median clustering  of polygonal curves and point sets under the (continuous and discrete) Fr\'{e}chet and Hausdorff distances respectively, if $C$ is defined by curves/point sets having at most $l$ vertices/points each, we know that $C$ is $h$-well behaved for $h(n) = (cnm/d)^{dl}$ (Lemmas~\ref{lem:cfd_well_behaved}, \ref{lem:dfd_well_behaved}, and \ref{lem:hd_well_behaved}) for some constant $c > 0$. Plugging this $h$ in Theorem~\ref{thm:cov} to get the size of the $\eps$-cover-code for $F_{k,C}$, we can show using some elementary algebra that $a = \Theta\left(\frac{klS^2m^{\delta}d}{\eps^2}\log\left( \frac{klS}{\eps}\right)\right)$ satisfies the requirements of Theorem~\ref{thm:cover_coreset} for any $\delta > 0$. Using the value of $S$ from Lemma~\ref{lem:sensitivity_metric_space}, we get the following result.

\begin{theorem}
\label{thm:coresets_existence}
For any instance of the $(k,C)$-median problem for  the metric spaces $(T_{poly}, \dist_F)$, $(T_{poly}, \dist_{dF})$, and $(U, \dist_H)$ with $C = T^l, U^l$ respectively, the algorithm in Fig.~\ref{fig:alg} computes a coreset (and the corresponding weights) of size $\Theta\left(\frac{k^3lm^{\delta}d}{\eps^2}\log\left( \frac{kl}{\eps}\right)\right)$ for any $\delta > 0$, with probability $\ge 1/2$; the constant inside $\Theta (\cdot)$ depends on $\delta$.
\end{theorem}

\subparagraph{Running time.} We briefly remark on the running time of the algorithm in Fig.~\ref{fig:alg}.

Apart from the time spent in the $(\alpha, \beta)$-approximation, computing the sensitivity upper bounds requires computing the distance between the $n$ input curves/point sets of $P$ (with at most $m$ points each) and the $\beta k$ centers of $C'$ (having at most $l$ points each). Hence there are $\beta n k$ distance computations, which take time $O(ml)$ each for $\dist_{dF}$ and $\dist_H$, and $O(ml \log ml)$ each for $\dist_F$~\cite{AG95}.

For $\dist_{dF}$ and $\dist_{H}$, a $\left(\frac{4}{3},1\right)$-approximation can be computed in times $O\left(nm \log m \log (m/l) \cdot 2^{O(k \log k)} \cdot l^{O(kdl)}\right)$ and $O\left(nm \cdot 2^{O(k \log k)} \cdot l^{O(kdl)}\right)$ respectively~\cite{NT20}.

For $\dist_F$, the algorithm in~\cite{BDR20} can be used to return a $(1.158,1)$-approximation in time $O\left(n \cdot \poly(m)\cdot2^{O(kdl)}\right)$. However, the centers returned can have $2l-2$ vertices and cannot be used directly to compute sensitivities. Using a 4-approximate minimum-error $l$-simplification on these center curves gives us curves with $l$ vertices, while only incurring a constant factor loss in the approximation factor (see, e.g.,~\cite[Theorem 3.1]{BDR20}). The simplification algorithm takes time $O(l^3 \log l)$ per curve~\cite[Lemma 7.1]{BDGHKLS19}. Combining all of these, we get the following.

\begin{corollary}
\label{thm:coresets_running_time}
The algorithm in Fig.~\ref{fig:alg} can be implemented in time $O\left(n \cdot \poly(m)\cdot2^{O(kdl)} + \frac{k^3lm^{\delta}d}{\eps^2}\log\left( \frac{kl}{\eps}\right) \right)$ for $\dist_F$; time $O\left( nm \log m \log (m/l) \cdot 2^{O(k \log k)} \cdot l^{O(kdl)} + \frac{k^3lm^{\delta}d}{\eps^2}\log\left( \frac{kl}{\eps}\right)\right)$ for $\dist_{dF}$; and time \\
$O\left( nm \cdot 2^{O(k \log k)} \cdot l^{O(kdl)} + \frac{k^3lm^{\delta}d}{\eps^2}\log\left( \frac{kl}{\eps}\right)\right)$ for $\dist_{H}$. The running times hold for any $\delta > 0$.
\end{corollary}

Note that these running times are illustrative only, and can be improved further by using faster bicriteria approximations.

\section{Lower bound on size of the coreset}
\label{sec:lower_bound}

We prove a lower bound on the size of a coreset, provided that the coreset is restricted to be a subset of the input points.

\begin{theorem}
\label{thm:lower_bound}
For any $c > 0$ and $\eps \in (0,1)$, there exists a set $P$ of polygonal curves (resp. finite point sets) in $\reals^2$ such that any $\eps$-coreset for the $k$-median problem on $P$ under the discrete or continuous Fr\'{e}chet (resp. Hausdorff) distance must have size $ \ge \frac{ck}{\eps}$, under the restriction that the coreset must be a subset of $P$.
\end{theorem}
\begin{proof}
We first provide proof for the continuous Fr\'echet distance -- the construction below also extends to the discrete Fr\'{e}chet and Hausdorff distances in a straightforward manner.

\begin{figure}[h]
\centering
\includegraphics[width=0.6\textwidth]{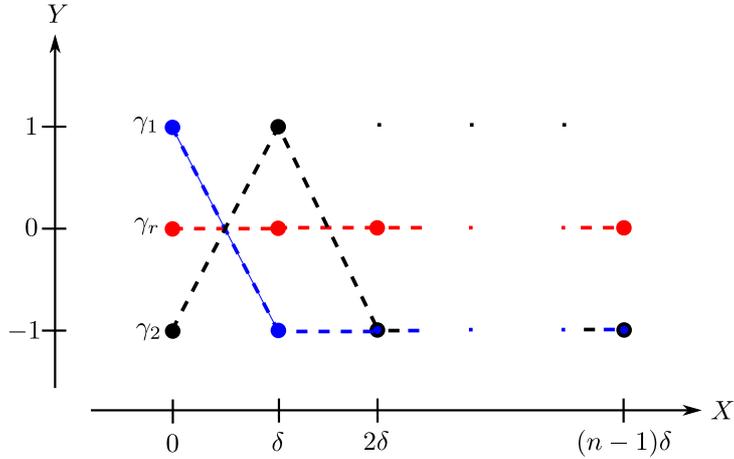}
\caption{Lower bound construction showing curves $\traj_r$ (red), $\traj_1$ (blue), and $\traj_2$ (black).}
\label{fig:lower_bound}
\end{figure}

Let $P = \{\traj_r, \traj_1, \ldots, \traj_n\}$ be a set of $n+1$ polygonal curves in $\reals^2$, for some $n > 0$ to be determined later. Let $\traj_r = \langle (0,0), (\delta,0), \ldots, ((n-1)\delta, 0)\rangle$. For $i = 1, \ldots, n$, define curve $\traj_i = \langle p_{i1}, \ldots, p_{in} \rangle$, where
\begin{align*}
p_{ij} = 
\begin{cases}
((j-1)\delta, -1) \text{ if } j \ne i. \\
((j-1)\delta, 1)  \text{ if } j = i.
\end{cases}
\end{align*}
If we take $\delta > 0$ large enough, the optimal correspondence determining the Fr\'{e}chet distance between any two curves in $P$ will match the $j$-th vertex of one with the $j$-th vertex of the other. Hence, it is clear that $\dist_F(\traj_r, \traj_i) = 1$ for all $i \in \{1, \ldots, n\}$, and $\dist_F(\traj_i, \traj_j) = 2$ for all $i, j \in \{1, \ldots, n\}, i \neq j$. See Fig.~\ref{fig:lower_bound}.

Suppose, to the contrary, that $P' \subseteq P$ is a coreset (with weight $w$) of size $< \frac{ck}{\eps}$ for some $c > 0$ and $\eps \in (0,1)$. Define $W = \sum_{x \in P', x \ne \traj_r} w(x) = \sum_{x \in P'} w(x) \dist_F(x, \traj_r)$; the last inequality stems from the construction and distance values.

Consider a $k$-median solution where all centers lie on $\traj_r$. We then have
\begin{align*}
W = \sum_{x \in P'} w(x) \dist_F(x, \traj_r) \ge (1-\eps) \cost (P, \traj_r) = (1-\eps) n,
\end{align*}
where the inequality is because $P'$ is an $\eps$-coreset. Let $C_1 \subseteq P \setminus (P' \cup \{\traj_r\})$ be of size $k$, and let $C_2$ be the $k$ curves in $P' \setminus \{\traj_r\}$ with largest weight. By an averaging argument we have
\begin{align*}
W' = \sum_{x \in C_2} w(x) \ge \frac{kW}{|P'|} > \frac{\eps W}{c} \ge \frac{\eps(1-\eps)n}{c}.
\end{align*}
Further,
\begin{align*}
\cost(P', C_1) &= \sum_{x \in P'} w(x) \dist_F(x, C_1) =  w(\traj_r) + 2W \le (1+\eps)\cost(P,C_1) \le 2(2n+1).\\
\cost(P', C_2) &= \sum_{x \in P'} w(x) \dist_F(x, C_2) = w(\traj_r) + 2(W-W') = \cost(P',C_1) - 2W'.
\end{align*}
Also, since $P'$ is an $\eps$-coreset, and $\cost(P, C_1) = \cost(P, C_2)$ by construction and our choice of $C_1$ and $C_2$, we have
\begin{align*}
\frac{\cost(P', C_1)}{\cost(P', C_2)} \le \frac{(1+\eps)\cost(P,C_1)}{(1-\eps)\cost(P,C_2)} = \frac{1+\eps}{1-\eps}.
\end{align*}
On the other hand,
\begin{align*}
\frac{\cost(P', C_1)}{\cost(P', C_2)}  = \frac{\cost(P', C_1)}{\cost(P', C_1) -2W'} > \frac{2(2n+1)}{2(2n+1) - \eps(1-\eps)n/c} > \frac{1+\eps}{1-\eps}
\end{align*}
for large enough $n$, which is a contradiction.
\end{proof}
\section{Conclusion}
\label{sec:conclusion}
We give the first coresets for $k$-median clustering polygonal curves and finite point sets under the Fr\'{e}chet and Hausdorff distances, whose sizes are independent of the number of input objects, when the cluster centers are restricted to have bounded complexity. In doing so, we precisely characterize a general condition on the restricted space of cluster centers that allows use to get small coresets. We also give a lower bound on the size of the coresets, when the coreset must be a subset of the input.

There are several interesting open problems. Can we extend our work to piecewise smooth curves (e.g., algebraic curves in the plane)? Can we use the work of Feldman and Langberg~\cite{FL11} to compute smaller coresets for our setting? As already mentioned, this requires a bound on the shattering dimension of range spaces induced by weighted distance functions. Can this be sidestepped using the probabilistic shattering dimension of Huang \etal~\cite{HJLW18}? As far as lower bounds are concerned, can we give tighter lower bounds ? The lower bound of $\Omega(k/\eps)$ in this paper uses a very general construction, and does not depend on $m,l$. We hope our work will spur further research in these directions.

\bibliography{ref}

\section{Appendix}
\subsection{Proof of Theorem~\ref{thm:cov}}
\label{sec:proof_cover}
We now show the existence of $G$ -- this will also prove Theorem~\ref{thm:cov}.

Let $Z = \sum_A 1/s(x)$, and $z(x) = 1/(s(x)Z)$ -- hence $z$ is a probability distribution on $A$.
Let $j = \arg \min_{f \in F_{k,C}} \hat{f}$, and let $\{v_1, \ldots, v_k\} \subseteq C$ be the associated centers of $j$. Let $\V_1, \ldots, \V_k$ be their corresponding Voronoi cells, i.e., for $x \in \V_i$, $j(x) = \dist(v_i, x)$. Observe that for any $f \in F_{k,C}$ and $y \in A$, we have
\begin{align}
\label{eq:7.5}
\hat{f} = \nu_A(f/s) \ge \frac{f(y)}{as(y)} \Rightarrow s(y) \ge \frac{f(y)}{a\hat{f}}.
\end{align}

For the following, we assume that $z(\V_i) > 0$ for all $i$ (we neglect all $i$ with $z(\V_i) = 0$). Define
\begin{align*}
j_i
= \frac{1}{az(\V_i)} \sum_{y \in \V_i \cap A} \frac{j(y)}{s(y)}
= \frac{1}{az(\V_i)} \sum_{y \in \V_i \cap A} \frac{\dist(v_i,y)}{s(y)}
= \frac{Z}{a}             \sum_{y \in \V_i \cap A} \frac{\dist(v_i,y)z(y)}{z(\V_i)}.
\end{align*}
This means that $\hat{j} = \sum_{i=1}^k z(\V_i)j_i$. Note that the average value (using distribution $z$) of $j$ on $\V_i$ is
\begin{align*}
\frac{1}{z(\V_i)} \sum_{y \in \V_i \cap A} \dist(v_i,y) z(y) = \frac{aj_i}{Z}.
\end{align*}
Hence by Markov's inequality, at least half of $\V_i \cap A$ (according to the distribution $z(\V_i)$) lies in $B(v_i, 2aj_i/Z)$.

Let $s_i^{\min} = \min_{y \in \V_i \cap A} s(y)$. Then
\begin{align}
\label{eq:7.6}
z(\V_i) = \sum_{\V_i \cap A} \frac{1}{s(x)Z} \ge \frac{1}{Zs_i^{\min}}
\end{align}

Let $f$ be any function in $F$, with centers $\{u_1, \ldots, u_k\}$. Let $\dist_i$ be the distance of $v_i$ to its nearest center of $f$. The next lemma bounds $\hat{f}$ from above and below.
\begin{lemma}
\label{lem:bound_hat_f}
$\hat{j} + \sum_{i=1}^k \dist_i / s_i^{\min} \ge \hat{f} \ge \frac{1}{4a} \sum_{i=1}^k \dist_i / s_i^{\min}$,  and hence
$\forall i, \hat{f} \ge\frac{\dist_i}{4a s_i^{\min}}$.
\end{lemma}
\begin{proof}
For a point $x \in A$, let $u_x$ be its closest center in $\{u_1, \ldots, u_k\}$.

Suppose that for all $i$ we have $\dist_i \ge 2aj_i/Z$. Let $B_i = \V_i \cap B(v_i, 2aj_i/Z) \cap A$. Then,
\begin{align*}
\hat{f}
&\ge \frac{1}{a} \sum_i \sum_{x \in B_i} \frac{f(x)}{s(x)} \\
&= \frac{1}{a} \sum_i \sum_{x \in B_i} \frac{\dist(x,u_x)}{s(x)} \\
&\ge \frac{1}{a} \sum_i \sum_{x \in B_i} \frac{\abs{\dist_i - 2aj_i/Z}}{s(x)} \\
&= \frac{Z}{a} \sum_i \sum_{x \in B_i} \abs{\dist_i - 2aj_i/Z} z(x) \\
&= \frac{Z}{a}  \sum_i \abs{\dist_i - 2aj_i/Z} \sum_{x \in B_i} z(x) \\
& \ge \frac{Z}{2a} \sum_i \abs{\dist_i - 2aj_i/Z}  z(\V_i) &\text{ (by Markov's inequality)} \\
& \ge \frac{Z}{2a} \sum_i \dist_i z(\V_i) -\sum_i j_i z(\V_i) \\
& \ge \frac{1}{2a} \sum_i \frac{\dist_i}{s_i^{\min}} - \sum_i j_i z(\V_i) &\text{ (by Equation~\ref{eq:7.6})} \\
&= \frac{1}{2a}\sum_i \frac{\dist_i}{s_i^{\min}} - \hat{j}.
\end{align*}
If for some $i$, it is the case that $\dist_i \le 2aj_i/Z$, then the inequalities above still hold since $\dist_i \le 2aj_i/Z \Rightarrow \dist_i - 2aj_i/Z \le 0 \le \dist(x,u_x)$.

Since $j$ was chosen so as to minimize $\hat{f}$ over all $f \in F_{k,C}$, we have $\hat{f} \ge \hat{j}$ which then gives $\hat{f} \ge \frac{1}{4a} \sum_i \dist_i / s_i^{\min}$.

For the upper bound, we have
\begin{align*}
\hat{f} 
&= \frac{1}{a} \sum_{x \in A} \frac{f(x)}{s(x)} = \frac{1}{a} \sum_{x \in A} \frac{\dist(x,u_x)}{s(x)} \\
&\le \frac{1}{a} \sum_i \sum_{x \in \V_i \cap A} \frac{\dist(x,v_i) + \dist_i}{s(x)} &\text{ (by triangle inequality)} \\
&\le  \frac{1}{a} \sum_i \sum_{x \in \V_i \cap A} \frac{\dist(x,v_i)}{s(x)} + 
 \frac{1}{a} \sum_i \sum_{x \in \V_i \cap A} \frac{\dist_i}{s(x)} \\
&\le \hat{j} + \sum_i \sum_{x \in \V_i \cap A} \frac{\dist_i}{as_i^{\min}} \le \hat{j} + \sum_i \sum_{x \in \V_i \cap A} \frac{\dist_i}{s_i^{\min}}.
\end{align*}
\end{proof}
We will now define a set of points in $X$ that will act as potential centers for functions $g \in G$. In what follows, we will use some parameters $p_1, p_2, p_3, p_4$ to be defined at the end of the proof. For each point $x \in A$ and for $i = 1,\ldots,(ap_1)^2$, let 
\begin{align*}
R_{x,i} 
&= \left\{ v \in X \mid \frac{\dist(x,v)}{s(x)} \in \left[ \frac{(i-1)\hat{j}}{ap_1}, \frac{i\hat{j}}{ap_1} \right)\right\} \\
&= B\left(x, \frac{i \hat{j} s(x)}{ap_1}\right) \setminus B\left(x, \frac{(i-1) \hat{j} s(x)}{ap_1}\right).
\end{align*}

Let $\EuScript{R} =\{ R_{x,i} \mid x \in A, i = 1, \ldots, (ap_1)^2\}$. Let $N \subseteq C$ have one element in each equivalence class of $\eqclass_{\EuScript{R}}^X$. Since $\X$ is $h$-well-behaved w.r.t. $C$, we have $|N| = h(|\EuScript{R}|) = h(a^3p_1^2)$.

The following lemma follows directly from the definition of $\EuScript{R}$.
\begin{lemma}
\label{lem:rep_eqclass}
Let $f \in F_{k,C}$, $E$ be an equivalence class in $\eqclass_{\EuScript{R}}^X$, and $n$ be the element of $N$ in $E$. Then, for any $v \in E$ and $x \in X$, we have
\begin{align*}
\abs{\dist(x,v) - \dist(x,n)} \le \frac{\hat{j}s(x)}{ap_1} \le \frac{\hat{f}s(x)}{ap_1}.
\end{align*}
\end{lemma}
We will now define the function $g$ that approximates $f$. We split our analysis into two cases.

\textbf{Case 1:} $\hat{f} \le p_1 \hat{j}$. We assume without loss of generality that each center $u_i$ of $f$ is significant in that it is the closest center of $f$ to some point $x \in A$. Otherwise, we can set the insignificant centers of $f$ to one of the significant centers  -- this will not change the value of $f$ at all, and the \emph{new} $f$ can be used in the analysis below.

Consider a center $u_i \in C$ of $f$ and let $x$ be a point in $A$ for which $u_i$ is the closest center of $f$ to $x$. By Equation~\ref{eq:7.5} we have 
\begin{align*}
\frac{\dist(x,u_i)}{s(x)} = \frac{f(x)}{s(x)} \le a \hat{f} \le ap_1 \hat{h}.
\end{align*}
Thus, $u_i \in \cup_i R_{x,i}$, which implies that $u_i$ is in some equivalence class of $\eqclass_{\EuScript{R}}^X$; let $n_i$ be the representative of $N$ in this equivalence class. We define $g$ to be the function in $F_{k,C}$ with centers $n_1,\ldots, n_k$.

We now show that $g$ satisfies the requirements of Equation~\ref{eq:g_req} with $c_f = 1$. Let $x \in A$ and consider any center $u_i$ of $f$ and its corresponding center $n_i$ of $g$. As $u_i$ and $n_i$ are in the same equivalence class in $\eqclass_{\EuScript{R}}^X$, it holds by Lemma~\ref{lem:rep_eqclass} that $\abs{\dist(x,u_i) - \dist(x,n_i)} \le \hat{j}s(x)/(ap_1) \le \hat{f}s(x)/(ap_1)$.

To bound $\abs{f(x) - g(x)}$, suppose the closest centers of $f$ and $g$ to $x$ are $u_i$ and $n_{i'}$ respectively. Then,
\begin{align*}
\abs{f(x) - g(x)}
&=\abs{\dist(x,u_i) - \dist(x,n_{i'})} \\
&\le \abs{\dist(x,u_i) - \dist(x,n_{i})} + \abs{\dist(x,n_i) - \dist(x,n_{i'})} \\
&\le \hat{f}s(x)/(ap_1) + \abs{\dist(x,n_i) - \dist(x,n_{i'})}.
\end{align*}
Now, by Lemma~\ref{lem:rep_eqclass} and the fact that $n_{i'}$ (resp. $u_i$) is the closest center of $g$ (resp. $f$) to $x$, it holds that
\begin{align*}
\dist(x,n_i) - 2 \hat{f}s(x)/(ap_1)
&\le \dist(x,u_i) - \hat{f}s(x)/(ap_1) \\
&\le \dist(x,u_{i'}) - \hat{f}s(x)/(ap_1) \\
&\le \dist(x,n_{i'}) \le \dist(x,n_i).
\end{align*}
Thus, $ \abs{\dist(x,n_i) - \dist(x, n_{i'})} \le 2\hat{f}s(x)/(ap_1)$. Hence, we conclude that $\abs{f(x) - g(x)} \le 3\hat{f}s(x)/(ap_1)$.
Setting $ap_1 \ge 768S/\eps$, we obtain $\abs{f(x) - g(x)} \le \frac{\eps}{256S}s(x)\hat{f}$. The total number of possible functions $g$ in this case is at most $|N|^k = (h(a^3p_1^2))^k$.

\textbf{Case 2:} $\hat{f} \ge p_1 \hat{j}$. In this case, the function $g$ we construct will be constant on each region $\V_i$ (recall that $\V_i$ is the Voronoi cell of center $v_i$ of $j$), the constant value being one of $O(p_2 \log p_2)$ different values we specify shortly. Thus, the number of different functions $g$ in this case is at most  $O(p_2^k (\log p_2)^k)$.

We now define $g$. Let $p_3$ be a parameter. An index $i$ is said to be \emph{light} if $\frac{\dist_i}{s_i^{\min}} \le \hat{f}/p_3$; else it is said to be \emph{heavy} (recall that $\dist_i$ is the distance of $v_i$ to its nearest center of $f$, and $s_i^{\min} = \min_{y \in \V_i \cap A} s(y)$). For light indices $i$, we set $g = 0$ on $\V_i$.

Define $c_f = \max_{i : i \text{ is heavy}} \frac{\dist_i}{p_2}$. For heavy $i$, i.e., for $i$ in which $\frac{\dist_i}{s_i^{\min}} \ge \hat{f}/p_3$, we define $g$ on $\V_i$ to be the nearest value to $\dist_i / c_f$ in the set
\begin{align*}
\left\{ 0, \left(1 + \frac{1}{p_2}\right),\left(1 + \frac{1}{p_2}\right)^2, \ldots, p_2 \right\}.
\end{align*}
Thus, the number of different values that $g$ can take is $O(p_2 \log(p_2))$. Note that by definition, $\dist_i / c_f \le p_2$ for heavy $i$. Also, for heavy $i$ and $g$ defined on $\V_i$ as mentioned, $c_f g \in \left[ \dist_i\left(1-\frac{1}{p_2}\right), \dist_i\left(1 + \frac{1}{p_2}\right)\right]$. We now bound $\abs{f(x) - c_fg(x)}$ separately for light and heavy indices.

\textbf{Light indices :} Recall that for such indices $i$ we have $\dist_i / s_i^{\min} \le \hat{f} / p_3$, and $g(x) = 0$ for $x \in \V_i$. Then for $x \in \V_i$,
\begin{align*}
\abs{f(x) - c_fg(x)} &= \abs{f(x)} \le \dist(x,v_i) + \dist_i \\
&= j(x) + \dist_i \\
&\le a \hat{j}s(x) + s_i^{\min} \hat{f}/p_3 &\text{   (by Equation~\ref{eq:7.5})}\\
&\le a s(x) \hat{f}/p_1 + s(x)\hat{f}/p_3 \\
& = \left(\frac{a}{p_1} + \frac{1}{p_3}\right)s(x)\hat{f}.
\end{align*}
By setting $p_1, p_2$ later so that $\left(\frac{a}{p_1} + \frac{1}{p_3}\right)\le \frac{\eps}{256S}$, we get $\abs{f(x) - c_fg(x)} \le \frac{\eps}{256S}s(x)\hat{f}$.

\textbf{Heavy indices :} Recall that $\dist_i / s_i^{\min} \ge \hat{f}/p_3$ for heavy indices $i$. Here we again consider two sub-cases.

Firstly, consider $x \in \V_i$ for which $\dist(x,v_i) \le \dist_i / p_4$ ($p_4 \ge 1$ will be set later). Since the closest center of $f$ to $v_i$ is at distance $\dist_i$ from $v_i$, for such $x$ we have $f(x) \in [\dist_i(1-1/p_4), \dist_i(1+1/p_4)]$. Thus
\begin{align*}
\abs{f(x) - c_fg(x)}
&\le \abs{f(x) - \dist_i} + \frac{\dist_i}{p_2} &\text{   (using value of $g$ for heavy indices)}\\
&\le \abs{\dist_i(1+1/p_4) - \dist_i} + \frac{\dist_i}{p_2} \\
&\le \dist_i \left( \frac{1}{p_4} + \frac{1}{p_2} \right) \\
&\le 4as_i^{\min} \hat{f}\left( \frac{1}{p_4} + \frac{1}{p_2} \right). &\text{   (using Lemma~\ref{lem:bound_hat_f})}
\end{align*}
Later, we will set $p_2,p_4$ such that $\left(\frac{1}{p_4} + \frac{1}{p_2} \right) \le \frac{\eps}{1024Sa}$, which will imply that  $\abs{f(x) - c_fg(x)} \le \frac{\eps}{256S}s(x)\hat{f}$.

Next, consider $x \in \V_i$ for which $\dist(x,v_i) > \dist_i/p_4$. By triangle inequality, we have $f(x) \le \dist_i + \dist(x,v_i)$. Also, by Equation~\ref{eq:7.5}, $j(x) = \dist(x,v_i) \le as(x)\hat{j}$. Then,
\begin{align*}
\abs{f(x) - c_fg(x)}
&\le \abs{f(x) - \dist_i} + \frac{\dist_i}{p_2} &\text{   (using value of $g$ for heavy indices)}\\
&\le \dist(x,v_i) + \frac{\dist_i}{p_2} \\
&< \dist(x,v_i)\left(1 + \frac{p_4}{p_2} \right)\\
&\le as(x)\hat{j}\left(1 + \frac{p_4}{p_2} \right) \\
&\le \frac{as(x)\hat{f}}{p_1}\left(1 + \frac{p_4}{p_2} \right). &\text{   (by definition of $\hat{j}$, and since $p_1 \ge 1$)}
\end{align*}
We will set $p_1,p_2,p_4$ such that $\frac{a}{p_1}\left(1 + \frac{p_4}{p_2} \right) \le \frac{\eps}{256S}$. This gives $\abs{f(x) - c_fg(x)} \le \frac{\eps}{256S}s(x)\hat{f}$.

To summarize, the following values satisfy all the requirements stated above: $p_1 = \left(\frac{Sa}{\eps}\right)^{\Theta(1)}, p_2 = \frac{\alpha_1Sa}{\eps}, p_3 = \frac{\alpha_2Sa}{\eps}, p_4 = \frac{\alpha_3Sa}{\eps}$ for suitable constants $\alpha_1,\alpha_2,\alpha_3$. Thus, the total size of $G$ is
\begin{align*}
(h(a^3p_1^2))^k + O(p_2^k(\log p_2)^k) = \left(h\left(\left( \frac{Sa}{\eps}\right)^{\Theta(1)}\right)\right)^k.
\end{align*}

\subsection{Proof of Lemma~\ref{lem:sensitivity_metric_space}}
\label{sec:sensitivity_proof}
\begin{proof}
Let $\mu$ be a distribution on $X$. Further, let $\{c^*_1, c^*_2, \ldots, c^*_k\} \subseteq C$ be such that $f^* = f_{c^*_1, \ldots, c^*_k}$ minimizes
\begin{align*}
\overline{f} = \int_X f d\mu
\end{align*}
over all $f$. Let $\overline{f^*} = \opt$, $\mu_i = \mu(\V(c^*_i))$ and $m_i = \frac{1}{\mu_i} \int_{\V(c^*_i)} f d\mu$. Thus, $\opt = \sum_i \mu_i m_i$, and $m_i$ is the \emph{average} distance of a point in $\V(c^*_i)$ (under $\mu$) from $c^*_i$. Each $\mu_i$ is positive unless the support of $\mu$ is less than $k$ in which case the theorem is trivial.  Using Markov's inequality we then have
\begin{align*}
\mu(\ball(c^*_i, 2m_i) \cap \V(c^*_i)) \ge \mu_i/2.
\end{align*}

Let $c_1, \ldots c_k$ be any set of $k$ points in $C$, and let $f = f_{c_1,\ldots,c_k}$. Let $c$ denote a closest point to $c^*_i$ in $\{c_1, \ldots, c_k\}$, and let $\dist_i = \dist(c,c^*_i)$. Then using the preceding inequality and the triangle inequality we have
\begin{align*}
\overline{f}
\ge \int\limits_{\V(c^*_i)} f d\mu
\ge  \int\limits_{\ball(c^*_i, 2m_i) \cap \V(c^*_i)} f d\mu
\ge \max \{0, \dist_i - 2m_i\} \mu_i/2.
\end{align*}

Also, $\overline{f} \ge \opt$. Thus, for any $\alpha \in [0,1]$ we have
\begin{align*}
\overline{f} \ge \max\{0, \dist_i - 2m_i\} \alpha\mu_i/2 + (1-\alpha)\opt.
\end{align*}
The value of $\alpha$ will be set later. Let $x \in \V(c^*_i)$.  We now have
\begin{align*}
\sens(x)
&= \max_f f(x) / \overline{f}\\
&= \max_f \dist(x,c) / \overline{f}\\
&\le \max_f \frac{\dist_i + \dist(x,c^*_i)}{\max\{0, \dist_i - 2m_i\} \alpha\mu_i/2 + (1-\alpha)\opt}\\
&\le \max_{\dist_i \ge 2m_i} \frac{\dist_i + \dist(x,c^*_i)}{\left(\dist_i - 2m_i\right) \alpha\mu_i/2 + (1-\alpha)\opt}.
\end{align*}
The right hand side is maximized either at $\dist_i = 2m_i$ or $\dist_i = \infty$. We conclude that
\begin{align*}
\sens(x)
\le \max \left\{ \frac{2m_i + \dist(x,c^*_i)}{(1-\alpha)\opt}, \frac{2}{\alpha \mu_i}\right\}
\le \frac{2m_i + \dist(x,c^*_i)}{(1-\alpha)\opt} + \frac{2}{\alpha \mu_i}.
\end{align*}
The total sensitivity can now be bounded by
\begin{align*}
\Sens(F_{k,C}) 
&= \sum_i \left( \int\limits_{\V(c^*_i)} \sens d\mu\right)\\
&\le \sum_i \left( \int\limits_{\V(c^*_i)} \left(  \frac{2m_i + \dist(x,c^*_i)}{(1-\alpha)\opt} + \frac{2}{\alpha \mu_i} \right)d\mu\right)\\
&= \sum_i \left( \frac{2m_i \mu_i}{(1-\alpha)\opt} + \frac{m_i \mu_i}{(1-\alpha)\opt} + \frac{2}{\alpha}\right)\\
&= \frac{3}{1-\alpha} + \frac{2k}{\alpha}\\
&\le (\sqrt{3} + \sqrt{2k})^2  \;\;\;\; \text{(minimum of $a/(1-\alpha) + b/\alpha$ is $(\sqrt{a} + \sqrt{b})^2$)}\\
&\le 4k + 6.
\end{align*}
\end{proof}

\end{document}